\documentclass[twoside]{article}
\usepackage{latexsym,amsmath,amssymb}
\usepackage{hyperref}
\usepackage{cite} 
\usepackage{sectsty}
  \sectionfont{\large} \subsectionfont{\normalsize}
\usepackage{graphicx}
\usepackage{theorem}
\usepackage[all]{xy}
\newtheorem{theorem}{Theorem}[section]
{\theorembodyfont\rmfamily
\newtheorem{lemma}{Lemma}[section]

\newtheorem{definition}{Definition}
}
\newenvironment{proof}{\begin{paragraph}
          {Proof}}{\end{paragraph}}

\renewenvironment{abstract}
 {\small\begin{quote}{\textbf{Abstract}}\,\,}{\end{quote}}
\newenvironment{keywords}
 {\small\begin{quote}{\textbf{Keywords}}\,\,}{\end{quote}}
\newenvironment{classification}
 {\small\begin{quote}{\textbf{2010 Mathematics Subject Classification}}\,\,}{\end{quote}}
\date{}
\title{\vspace{-9ex}
{\centering
 \textbf{\large Pricing Formulae of Power Binary and Normal Distribution Standard Options and  Applications
}}}

 \vspace{-4ex}

\author{\small\textsf{\bfseries 
$^{1}$~~ Hyong-Chol O, $^{2}$~~Dae-Sung Choe}\\[-.5ex]
{\footnotesize  ${}^{1, 2}$ Faculty of Mathematics, \textbf{Kim Il Sung} University,}
{\footnotesize   Pyongyang , D P R K}\\[-.5ex]
{\footnotesize e-mail: $^{1}$hc.o@ryongnamsan.edu.kp }}

\begin{document}

\maketitle
\thispagestyle{empty}

\vspace{-.6cm}

\begin{abstract}

In this paper the Buchen's pricing formulae of (higher order) asset and bond binary options are incorporated into the pricing formula of power binary options and a pricing formula of "the normal distribution standard options" with the maturity payoff related to a power function and the density function of  normal distribution is derived. And as their applications, pricing formulae of savings plans that provide a choice of indexing and discrete geometric average Asian options are derived and the fact that the price of discrete geometric average Asian option converges to the price of continuous geometric average Asian option when the largest distance between neighboring monitoring times goes to zero is proved. 
\end{abstract}

\begin{keywords}
power binary option; normal distribution standard option; savings plans that provide a choice of indexing; discrete geometric average Asian option
\end{keywords} 


%
%%%%%% Introduction %%%%%%%%%%% 
%

\section{Introduction}

\indent

Option pricing problem is one of the main questions in financial mathematics. The studying method expressing complex exotics or corporate bonds as a portfolio of binaries becomes a widely used way in solving financial practical problem \cite{P1,P2,I,RR,OK,ODC,ODJ,OYD,OJJ,OJS,SB}.

In this paper, solution representations of Black-Scholes PDE with power functions or cumulative distribution functions of normal distribution as a maturity payoff are provided and using them savings plans that provide a choice of indexing and discrete geometric average Asian options are studied. This becomes binary option methods and their applications. 

The basic idea of expressing the payoffs of complex options in terms of binary options was initiated in Rubinstein and Reiner(1991) \cite{RR}, where they considered the relationship of barrier options and binaries. Ingersoll(2000) \cite{I} extended the idea by expressing complex derivatives in terms of “event- driven” binaries. Event $\varepsilon$ driven binary option pays one unit of underlying asset if and only if the event $\varepsilon$ occurs, otherwise it pays nothing.

Buchen(2004) introduced the concepts of first and second order asset and bond binary options and provided the pricing formula for them using expectation method. And then he applied them to pricing some dual expiry exotic options including compound option, chooser option, one time extendable option, one time shout option, American call option with one time dividend and partial barrier options. 

Skipper(2009) \cite{SB} introduced a general concept of multi-period, multi-asset M-binary which contains asset and bond binaries of Buchen \cite{P1} and provided the pricing formula by expectation method. They noted that portfolios of M-binaries can be used to statically replicate many European exotics and applied their results to pricing a simple model of an executive stock option (ESO) which has been used in practice to value corporate remuneration packages.  

In \cite{OK}, the authors extended the concepts of Buchen's first and second order asset and bond binary options to a concept of higher order binary option, deduced the pricing formulae by PDE method and use it to get the pricing of some multiple-expiry exotic options such as Bermudan option and multi time extendable option and etc.

In \cite{ODJ} and \cite{OJS} the authors obtained some pricing formulae of corporate bonds with credit risk including the discrete coupon bonds using the method of higher order binary options. In \cite{OYD} the authors extend higher order binary options into the case with time dependent coefficients and using it studied the discrete coupon bonds. 

Some banks have a service of savings account with choice item of interest rates, the holder of which has the right to select one among domestic and foreign interest rates and \cite{BB} calculated the   price of such savings plans that provide a choice of indexing using expectation method. In \cite{ORW} the authors studied the price of such savings plans that provide a choice of indexing using PDE method. In such savings plans, the foreign currency price of the expiry paypff is related to the the inverse of underlying asset price \cite{ORW}.

The geometric average Asian option is one of well studied exotics. In \cite{J} they studied continuous geometric average Asian options using PDE method and provided the pricing formulae, call-put parities, binomial tree models and explicit difference scheme on characteristic lines. In \cite{P2} they studied continuous models of geometric average Asian options with zero dividend rates using expectation and PDE methods and provided the pricing formulae of discrete geometric average Asian options with expectation method. They showed that the prices of discrete geometric average Asian options converge to the price of the corresponding continuous geometric average Asian options when the monitoring time interval goes to zero. The expiry payoffs of discrete geometric average Asian options are related to the n-th root of the price of underlying asset.

The purpose of this article is to study the pricing model of the option whose terminal payoff is related to the power of the price of the underlying asset in the viewpoint of PDE and apply them into pricing problems of some financial contracts.

First we find the pricing formula of so called~ "power binary option", the expiry payoff of which is the power of underlying asset price, by using the partial differential equation method. Our formula includes the pricing formulae of \cite{P1, OK} for the (higher order) asset and bond binary options. We apply it to calculate the price of savings plans that provide a choice of indexing.

Then we derive the solution representation of the terminal value problem of the Black-Scholes partial differential equation whose terminal value is the product of a power function and a normal distribution function and using it we obtain the pricing formula of discrete geometric average Asian option. And we prove that the price of discrete geometric average Asian option converges to the price of the continuous geometric average Asian option as the interval between neighboring monitoring times goes to zero. It is very interesting in the view of PDE that the limit of the solutions to connecing problems of initial value problems of Black-Scholes partial differential equations on several time intervals becomes the solution to the initial value problem of another partial differential equation. This shows that using the backgrounds of problems a solving problem of a complicated differential equation can be reduced to solving problems of simple differential equations.

The rest of this paper is organized as follow. In section 2 we provide the pricing formula of an European option (power binary option and normal distribution vanilla option) whose payoff is related to the power of the underlying asset price and density function of the normal distribution by using the partial differential equation method. In section 3 we apply the result of section 2 to provide the pricing formula of savings plans that provide a choice of indexing and discrete geometric average Asian option and we describe the result of the convergence of the price of discrete geometric average Asian option as the interval between neighboring monitoring times goes to zero. In section 4 we povide the price formula of high order power binary option. 

%%%%%%%2. Pricing Formulae of Power Binary and Normal Distribution Standard Options%%%%%
\section{Pricing Formulae of Power Binary and Normal Distribution Standard Options}

\begin{definition}
 ~Let $r,q,\sigma $~be interest rate, the dividend rate and the volatility of the underlying asset respectively. Consider the terminal-value problem of the following Black-Scholes equation.
\end{definition}
\begin{equation}\label{eq-1}
\frac{\partial V}{\partial t}+\frac{1}{2}\sigma ^2 x^2 \frac{\partial V}{\partial x^2}+(r-q)x\frac{\partial V}{\partial x}-rV=0,(0<t<T,x>0)
\end{equation}
\begin{equation}\label{eq-2}
V(x,T)=x^\alpha
\end{equation}
According to preposition 1 of \cite{OK}, the solution of this equation is given as following:
\begin{equation}\label{eq-3}
V(x,t)=\frac{e^{-r(T-t)}}{\sigma \sqrt{2 \pi (T-t)}}\int_0^{\infty} \frac{1}{z}e^{-\frac{1}{2\sigma ^2(T-t)}\left(\ln\frac{x}{z}+\left(r-q-\frac{\sigma ^2}{2}\right)(T-t)\right) ^2} z^{\alpha}dz
\end{equation}

This $V(x,t)$ is called a \textbf{"standard power option" } or \textbf{"standard  $\alpha$-power option"} with expiry payoff $x^ \alpha$.~We denote the price of standard  $\alpha$-power option as $M^ \alpha (x,t)$.

The pricing formula of standard power option is provided as following.

\begin{theorem}\label{theorem1}
The price of $\alpha$-power standard option, a solution of the Black-Scholes equation \eqref{eq-1}, \eqref{eq-2}  is given as 
\end{theorem}
\begin{equation}\label{eq-4}
M^{\alpha}(x,t)=e^{\mu (T-t)}x^{\alpha}
\end{equation}

Here

\begin{equation}\label{eq-5}
\mu (r,q,\sigma ,\alpha)=(\alpha -1)r-\alpha q+\frac{\sigma ^2}{2}(\alpha ^2-\alpha)
\end{equation}
\begin{proof}
 Using the transformation

$$y=\left[\ln \frac{x}{z}+\left(r-q-\frac{\sigma ^2}{2}+\alpha \sigma ^2 \right)(T-t) \right]\left( \sigma \sqrt {(T-t)} \right )^{-1}$$

Then the equation \eqref{eq-3} is changed into the following.

$$V(x,t)=e^{\mu (T-t)}x^{\alpha} \frac{1}{\sqrt {2\pi}}\int_{-\infty}^{\infty} e^{-\frac{y^2}{2}}dy=e^{\mu (T-t)x^{\alpha}}.$$
~~~~~~~~~~~~~~~~~~~~~~~~~~~~~~~~~~~~~~~~~~~~~~~~~~~~~~~~~~~~~~~~~~~~~~~~~~~~~~~~~~~~~~~~~(Q.E.D) 
\end{proof}

\begin{definition} 
The binary contract based on the standard  $\alpha$-power option is called ''$\alpha$-power binary option''. i.e the price of a power binary option is a solution of \eqref{eq-1} satisfying the terminal value condition 
\end{definition}
\begin{equation}\label{eq-6}
V(x,t)=x^{\alpha}\cdot 1(sx>s\xi)
\end{equation}
Here s(+ or -) is called ''sign indicator'' of upper or lower binary options and we denote the price of a  $\alpha$-power option as $\left(M^\alpha \right)_\xi ^s(x,t)$ .Since

$$\left (M^\alpha \right )_\xi ^{+}(x,T)+\left (M^\alpha \right )_\xi ^{-}(x,T)=M^\alpha (x,T)$$

we have the following symmetric relation between the  $\alpha$-power standard option price and the corresponding upper  $\alpha$-power binary option and lower  $\alpha$-power binary option.
\begin{equation}\label{eq-7}
\left (M^\alpha \right )_\xi ^{+}(x,t)+\left (M^\alpha \right )_\xi ^{-}(x,t)=M^\alpha (x,t)
\end{equation}

\begin{theorem}\label{theorem 2}
The price of  $\alpha$-power binary option (solution of \eqref{eq-1},\eqref{eq-5}) is given as following.
\begin{equation}\label{eq-8}
\left(M^\alpha \right)_\xi ^s(x,t)=e^{\mu (T-t)}x^{\alpha}N(sd)
\end{equation}
Here
\begin{eqnarray}\label{eq-9}
d & =& d(\frac{x}{\xi},r,q,\sigma ,\alpha ,T-t)\nonumber
\\
& =& \left[\ln \frac{x}{\xi}+\left (r-q-\frac{\sigma ^2}{2}+\alpha \sigma ^2 \right )(T-t)\right](\sigma \sqrt {T-t})^{-1} 
\end{eqnarray}
$$N(d)=\frac {1}{\sqrt {2 \pi}} \int_{-\infty}^d e^{-\frac{y^2}{2}}dy$$
\end{theorem}

\begin{proof}
From the preposition 1 of \cite{OK}, we have

$${\displaystyle{\left(M^\alpha \right)_\xi ^s(x,t)=\frac{e^{-r(T-t)}}{\sigma \sqrt {2\pi (T-t)}}\int_0^{\infty}\frac{1}{z}e^{\frac{1}{2\sigma ^2(T-t)}\left [ \ln \frac{x}{z}+\left( r-q-\frac{\sigma ^2}{2}\right)^2(T-t)\right]^2}z^\alpha 1(sz>s\xi)dz}}$$
Let's use the variable transformation
$$y=\left[ \ln \frac{x}{z}+\left( r-q+\frac{\sigma ^2}{2}+\alpha \sigma ^2\right)(T-t)\right](\sigma \sqrt {(T-t)})^{-1}$$
Then $$\ln z=\left[ \ln x+\left( r-q-\frac {\sigma ^2}{2}+\alpha \sigma ^2\right)(T-t)\right]-y(\sigma \sqrt {(T-t)})$$~and consider
$$1(sz>s\xi)=1(s\ln z>s\ln \xi)=1(sy<sd)$$
Then we have
$$\left(M^\alpha \right)_\xi ^{+}(x,t)=e^{\mu (T-t)}x^\alpha \frac{1}{\sqrt {2\pi}}\int_{-\infty}^{+\infty}e^{-\frac{-y^2}{2}}1(sy<sd)dy=I$$
Again using the variable transformation $y'=sy,dy'=sdy$ we have
\begin{eqnarray}
I &=& e^{\mu (T-t)}x^\alpha \frac{1}{\sqrt {2\pi}}s\int_{-s\infty}^{+s\infty}e^{-\frac{-{y'}^2}{2}}1(y'<sd)dy'=e^{\mu (T-t)}x^\alpha \frac{1}{\sqrt {2\pi}}\int_{-\infty}^{sd}e^{-\frac{{-y'}^2}{2}}dy'=\nonumber
\\
&=& e^{\mu (T-t)}x^\alpha N(sd)\nonumber
\end{eqnarray}
(QED)
\end{proof}

\textbf{Remark 1:} Consider the case when the binary condition of the expiry payoff function has more general form. i.e

\begin{equation}\label{eq-10}
V(x,t)=x^ \alpha \cdot 1(sx^\beta >s\xi) 
\end{equation}
We have
$$1(sx^\beta >s\xi)=1(s \cdot sgn(\beta)x>s \cdot sgn(\beta)\xi ^{\frac{1}{\beta}})=1(tx>t\zeta)$$
So if we $t=s \cdot sgn(\beta),\zeta =\xi ^{\frac{1}{\beta}}$, the solution of \eqref{eq-1},\eqref{eq-10} is given as $\left(M^\alpha \right)_\zeta ^t(x,t)$.

\textbf{Remark 2:} When $\alpha =0$ , the  $\alpha$-power binary option is the cash binary option and when $\alpha =1$ , it is the asset binary option. Theorem \eqref{theorem 2} includes the results for the first order binary options of \cite{P1,P2,J} as a special case. And as shown in section 3, by using the  $\alpha $-power binary options we can represent the prices of financial contracts which are not able to be represented in terms of cash or asset binary options.

\begin{definition}
 The solution of \eqref{eq-1} with the expiry payoff 

\begin{eqnarray}\label{eq-11}
V(X,T)=X^\beta N\left[ \delta \left( \frac{X^i}{K},\tau _1,\tau '_1,\alpha \right)\right]
\end{eqnarray}

$$\delta \left( \frac{X^i}{K},\tau _1,\tau '_1,\alpha \right)=\left [ \ln \frac{X^i}{K}+\left( r-q-\frac{\sigma ^2}{2}+\alpha \sigma ^2\right)\tau _1\right ](\sigma \sqrt {\tau_1'})^{-1}$$

is called a "price of power normal distribution standard options". 
\end{definition}
Denote $\tau =T-t$.
\begin{theorem}\label{theorem 3}
 The solution of \eqref{eq-1},\eqref{eq-11} is provided as follows

\begin{equation}\label{eq-12}
V(X,\tau ;\tau _1,\tau _1')=X^{\beta}e^{\mu (\beta)\tau}N(d_1)
\end{equation}
where
\begin{eqnarray}\label{eq-13}
d_1 &=& \left[ \ln \frac{X^i}{K}+\left( r-q-\frac{\sigma ^2}{2}+\frac{i\beta \tau +\alpha \tau _1}{i\tau +\tau _1}\sigma ^2\right)(i\tau +\tau _1)\right](\sigma \sqrt {i^2\tau +\tau _1'})^{-1}=\nonumber
\\
&=& \delta \left( \frac{X^i}{K},i\tau +\tau_1,i^2\tau +\tau _1',\frac{i\beta \tau +\alpha \tau _1}{i\tau +\tau _1}\right)
\end{eqnarray}
and $\mu (\beta)=(\beta -1)r-\beta q+\displaystyle{\frac{\sigma ^2}{2}}(\beta ^2-{\beta})$ can be given by \eqref{eq-5}
\end{theorem}
\begin{proof}
From proposition 1 of \cite{OK}, the solution of \eqref{eq-1},\eqref{eq-11}  is expressed as follows:

\begin{multline}\label{eq-14}
V(X,\tau ,\tau _1,\tau _1')=
\\
\displaystyle{=\frac{e^{-r\tau}}{\sigma \sqrt {2\pi \tau}}\int_0^{+\infty}\frac{z^{\beta}}{z}\left(\frac{1}{\sqrt {2\pi}}\int_{-\infty}^{\delta (\frac{z^i}{K},\tau _1,\tau _1',\alpha)}e^{\frac{-y^2}{2}}dy \right)e^{-\frac{1}{2\sigma ^2\tau}\left(\ln \frac{X}{z}+(r-q-\frac{\sigma ^2}{2})\tau \right)^2}dz=I}
\end{multline}
Here

$$\delta (\frac{z^i}{K},\tau _1,\tau _1',\alpha)=\left[\ln \frac{z^i}{K}+(r-q-\frac{\sigma ^2}{2}+\alpha \sigma ^2)\tau _1 \right](\sigma \sqrt {\tau _1'})^{-1}$$
In \eqref{eq-14} take the variable transformation:
\begin{equation}\label{eq-15}
y_1=\frac{y\sqrt {\tau _1'}}{\sqrt {\tau _1'+i^2\tau}}+i\left[\ln \frac{X}{z}+(r-q-\frac{\sigma ^2}{2}+\beta \sigma ^2)\tau  \right](\sigma \sqrt {\tau _1'+i^2\tau})^{-1}
\end{equation}
$$y_2=\left[ \ln \frac{X}{z}+(r-q-\frac{\sigma ^2}{2}+\beta \sigma ^2)\tau \right]$$

Then we have

$$\left | \frac{\partial (y_1,y_2)}{\partial (y,z)} \right | =\frac{1}{z\sigma \sqrt \tau}\frac{\sqrt {\tau_1'}}{\sqrt {\tau_1'+i^2 \tau}}$$
And $z>0\Leftrightarrow y_2<+\infty$ and

$$y<\delta (\frac{z^i}{K},\tau _1,\tau _1',\alpha)\Leftrightarrow$$
$$\Leftrightarrow y_1<\frac{1}{\sigma \sqrt {\tau _1'+i^2\tau}}\left[ \ln \frac{X^i}{K}+(r-q-\frac{\sigma ^2}{2})(\tau _1+i\tau)+(\alpha \tau _1+i\beta \tau)\sigma ^2=\right]$$
$$=\delta \left( \ln \frac{X^i}{K},\tau _1+i\tau,\tau _1'+i^2\tau ,\frac{\alpha \tau _1+i\beta \tau}{\tau _1+i\tau}\right)=d_1$$

In the exponent of the last exponential function in the integrand of \eqref{eq-14}, we consider
$$exp\left( \beta\ln \frac{X}{z}\right)=\frac{X^\beta}{z^\beta}$$
$$e^{-r\tau}\cdot e^{\left(r\beta-q\beta+\frac{1}{2}\sigma ^2(\beta ^2-\beta) \right)\tau}$$
$$y=\frac{\sqrt {\tau _1'+i^2\tau}}{\sqrt {\tau _1'}}y_1-y_2i\frac{\sqrt \tau}{\sqrt {\tau _1'}}$$
Then we can rewrite equation \eqref{eq-14} by follows:

\begin{equation}\label{eq-16}
\displaystyle{I=\frac{X^\beta e^{\mu (\beta \tau)}}{2\pi}\frac{\sqrt {\tau _1'+i^2\tau}}{\sqrt {\tau _1'}}\int_{-\infty}^{d_1} \int_{-\infty}^{+\infty}e^{\frac{\left( \frac{\sqrt {\tau _1'+i^2\tau}}{\sqrt {\tau _1'}}y_1-y_2i\frac{\sqrt \tau}{\sqrt {\tau _1'}}\right)^2+y_2^2}{2}}}dy_2dy_1
\end{equation}
In \eqref{eq-16} the exponent is

$$\left( \frac{\sqrt {\tau _1'+i^2\tau}}{\sqrt {\tau _1'}}y_1-y_2i\frac{\sqrt \tau}{\sqrt {\tau _1'}}\right)^2+y_2^2=\left( \frac{\sqrt {\tau _1'+i^2\tau}}{\sqrt {\tau _1'}}y_2-y_1i\frac{\sqrt \tau}{\sqrt {\tau _1'}}\right)^2+y_1^2$$
Since $\frac{\sqrt {\tau _1'+i^2\tau}}{\sqrt {\tau _1'}}\int_{-\infty}^{+\infty}e^{\left( \frac{\sqrt {\tau _1'+i^2\tau}}{\sqrt {\tau _1'}}y_2-y_1i\frac{\sqrt \tau}{\sqrt {\tau _1'}}\right)^2}dy_2=\sqrt {2\pi}$, we can rewrite equation \eqref{eq-16} as follows:
$$I=\frac{X^\beta e^{\mu (\beta)\tau}\tau}{\sqrt {2\pi}}\int_{-\infty}^{d_1}e^{-\frac{y_1^2}{2}}dy_1=X^\beta e^{\mu (\beta)\tau} N(d_1).~~~~~~~~~~~~~~~~~~~~~~~~~~~~~~~~~~~(QED)$$
\end{proof}

%%%%%%%%3.Applications%%%%%%%%%
\section{Applications}

\subsection{Savings Plans that provide a choice of indexing}

In the savings plans with choice item of interest rates, the holder has the right to select one among the domestic rate and foreign rate. Thus this savings contract is an interest rate exchange option \cite{BB,ORW}.

Let $r_d$ denote the domestic (risk free) interest rate,  $r_f$ the foreign (risk free) interest rate,  $X(t)$ domestic currency / foreign currency exchange rate, $Y(t) = X(t)^{-1}$: foreign currency / domestic currency exchange rate.

\textbf{Assumption}

1) The invested quantity of money to the savings plan is 1 unit of (domestic) currency, $r_d,r_f$  are both constants. 

2) The maturity payoff is as follows:

$$V_f=V_d\cdot X^{-1}(T)=\max{e^{rd^T}X^{-1}(0)e^{r_fT}}(unit ~of~ foreign~ currency)$$
3)The exchange rate   satisfies the model of Garman-Kohlhagen \cite{MS}:

$$\frac{dX(t)}{X(t)}=\lambda \cdot dt+\sigma dW(t).$$
4) The foreign price of the option is given by a deterministic function $V_f=V_f(X,t)$.
Under the assumptions 1)-4),$V_f=V_f(X,t)$  , the foreign price function of the Savings Plans that provide a choice of indexing is the solution of the following initial value problem of PDE
\begin{equation}\label{eq-17}
\frac{\partial V_f}{\partial t}+\frac{1}{2}X^2\frac{\partial ^2V_f}{\partial X^2}+\left(\sigma ^2+r_d-r_f \right)X\frac{\partial V_f}{\partial X}-r_fV_f=0
\end{equation}
\begin{equation}\label{eq-18}
V_f(X,t)=\max \{e^{r_dT}X^{-1}(T),X^{-1}(0)e^{r_fT}\}.
\end{equation}
In fact, by Itô formula and assumption 4),
$$d\left(\frac{1}{X(t)} \right)=(-\mu X^{-1}(t)+\sigma ^2X^{-1}(t))dt-\sigma X^{-1}(t)dW(t)$$

By $\Delta$-hedging, construct a portfolio $\Pi$
$$\Pi=V_f(X,t)-\Delta \cdot X^{-1}$$
(This portfolio consists of a sheet of option, $\Delta$ units of domestic currency and its price is calculated in foreign currency.) We select $\Delta$ so that  $d\Pi =r_f dt$. By Itô formula, we have
$$dV_f=\left(\frac{\partial V_f}{\partial t}+\frac{1}{2}\sigma ^2 X^2\frac{\partial ^2V_f}{\partial X^2}+\mu X\frac{\partial V_f}{\partial X}\right)dt+\sigma X\frac{\partial V_f}{\partial X}dW.$$
Then

\begin{eqnarray}
d\Pi &=& dV_f-\Delta\left(\frac{1}{X}\right)-\Delta \cdot r_dX^{-1}dt=\nonumber
\\
&=& \left( \frac{\partial V_f}{\partial t}+\frac{1}{2}\sigma ^2 X^2\frac{\partial ^2V_f}{\partial X^2}+\mu X\frac{\partial V_f}{\partial X}\right)+\sigma X\frac{\partial V_f}{\partial X}dW-\Delta \cdot (-\mu X^{-1}+\sigma ^2X^{-1})\cdot dt-\nonumber
\\
& ~~~& -\Delta \sigma X^{-1}dW-\Delta r_dX^{-1}dt\nonumber
\\
&=& \left( \frac{\partial V_f}{\partial t}+\frac{1}{2}\sigma ^2 X^2\frac{\partial ^2V_f}{\partial X^2}+\mu X\frac{\partial V_f}{\partial X}-\Delta \cdot (-\mu X^{-1}+\sigma ^2X^{-1})-\Delta r_dX^{-1}\right)dt+\nonumber
\\
& ~~~& +\left(\sigma X\frac{\partial V_f}{\partial X}-\Delta \sigma X^{-1}\right)dW\nonumber
\\
&=& r_f\Pi dt=r_f(V_f-\Delta \cdot X^{-1})\cdot dt\nonumber
\end{eqnarray}
Thus
\begin{multline}
\left( \frac{\partial V_f}{\partial t}+\frac{1}{2}\sigma ^2 X^2\frac{\partial ^2V_f}{\partial X^2}+\mu X\frac{\partial V_f}{\partial X}-\Delta \cdot (-\mu X^{-1}+\sigma ^2X^{-1})-\Delta r_dX^{-1}\right)dt+\nonumber
\\
+\left(\sigma X\frac{\partial V_f}{\partial X}-\Delta \sigma X^{-1}\right)dW=0
\end{multline}
If we select $\Delta =X^2\frac{\partial V_f}{\partial X}$, then we have
$$\frac{\partial V_f}{\partial t}+\frac{1}{2}\sigma ^2 X^2\frac{\partial ^2V_f}{\partial X^2}+\mu X\frac{\partial V_f}{\partial X}-X^2\frac{\partial V_f}{\partial X}(\mu X^{-1}-\sigma ^2X^{-1}- r_dX^{-1}+r_fX^{-1})-r_fV_f=0$$
Thus we have the PDE model of option price.

$$\frac{\partial V_f}{\partial t}+\frac{1}{2}\sigma ^2 X^2\frac{\partial ^2V_f}{\partial X^2}+(\sigma ^2+r_d-r_f)X\frac{\partial V}{\partial X}-r_fV_f=0$$
\eqref{eq-17} is a BS-PDE with rate $r_f$ , dividend $2r_f-r_d-\sigma ^2$ and volatility $\sigma $. We can represent the solution of the problem \eqref{eq-17},\eqref{eq-18} in terms of the power binary options. $x_0=X(0)$  is the known quantity and the maturity payoff function \eqref{eq-18} is rewritten as follows.

\begin{eqnarray}
V_f(X,T) &=& \max \{e^{r_dT}X^{-1},x_0^{-1}e^{r_fT}\}=\nonumber
\\
&= & e^{r_dT}X^{-1}\cdot 1\left(e^{r_dTX^{-1}<x_0^{-1}e^{r_dT}}\right)\nonumber
\\
&=& e^{r_dT}X^{-1}\cdot 1(X<K)+x_0^{-1}\cdot 1(X>K).\nonumber
\end{eqnarray}
Here $K=x_0e^{(r_d-r_f)T}.$ Therefore by theorem 2 the option's foreign price is calculated as follow.
$$V_f(x,t)=e^{r_dT}(M^{-1})_K^{-}(x,t)+x_0^{-1}e^{r_fT}(M^0)_K^{+}(x,t)$$
From \eqref{eq-5}, we obtain the following theorem. (we considered $\mu (-1)=-r_d,\mu(0)=-r_f$)
\begin{theorem}
Under the assumptions 1~ 4, the foreign price of the savings plans that provide a choice of indexing is the solution of the problem \eqref{eq-17},\eqref{eq-18} and given by
\begin{equation}\label{eq-19}
V_f(x,t)=e^{r_dt}X^{-1}N(-d_1)+x_0^{-1}e^{r_ft}N(d_2)
\end{equation}

Here

\begin{equation}\label{eq-20}
d_1=\frac{\displaystyle{\ln \frac{x_0^{-1}e^{r_ft}}{e^{r_dtX^{-1}}}}-\frac{\sigma ^2}{2}(T-t)}{\sigma \sqrt{T-t}},d_2=\frac{\displaystyle{\ln \frac{x_0^{-1}e^{r_ft}}{e^{r_dtX^{-1}}}}+\frac{\sigma ^2}{2}(T-t)}{\sigma \sqrt{T-t}}
\end{equation}

\end{theorem}

\textbf{Remark.} The financial meaning of this formula is clear.  $e^{r_dT}X^{-1}$ is the foreign price of the present value of the savings account when 1 unit (domestic currency) is saved with the domestic rate and  $x_0^{-1}e^{r_ft}$ is the present value (foreign currency) of the savings account when 1 unit (domestic currency) is saved  with foreign rate.  $N(-d_1)$ and $N(d_2)$  are the proportion of these two quantities in the option price. This proportion depends on which is bigger among them and especially at expiry date one is 0 and the other is 1. The formula \eqref{eq-19},\eqref{eq-20} is equal to that of \cite{ORW} when the inflation rate is 1.

\subsection{Geometric average Asian option}

\indent

The geometric average Asian option is an option whose expiry payoff depends on not only the terminal price of the underlying asset but also the geometrical average of the underlying asset's prices on its life time \cite{J}. The price of this option at time t depends on the underlying asset's price at time t as well as the geometrical average $J_t$ of the underlying asset's prices on the time interval $[0,t]$.

Let  $0=T_1<T_2<\cdots <T_m=T$ be discrete monitoring times and $X_n$   be the price of the underlying asset. Then the discrete geometric average of the underlying asset's prices till the time  $T_n$ is expressed by
$$J_n=J_{T_n}=\left( \prod_{i=1}^{n} X_i\right)^{\frac{1}{n}}=e^{\frac{1}{n} \sum_{i=1}^n{\ln X_i}}$$
The option with expiry payoff
$$(J_T-K)^{+}=(\sqrt[m]{X_1 \cdots X_m}-K)^{+}$$
is called an \textbf{discrete geometrical average Asian option}.

Similarly the geometrical average of the underlying asset's price to time $[0,t]$ is expressed as 
$$J_t=e^{\frac{1}{t}\int_0^t{\ln S_\tau d\tau}}$$ 
The option with expiry payoff
$$(J_T-K)^{+}=\left(e^{\frac{1}{t}\int_0^t{\ln S_\tau d\tau}}-K\right)^{+}$$
is called a \textbf{continuous geometric average Asian (call) option(with fixed exercise price)}. The price of the continuous geometric average Asian (call) option (with fixed exercise price) is the solution of the following problem 

\begin{eqnarray}\label{eq-21}
&\displaystyle{\frac{\partial V}{\partial t} +J\frac{\ln X-\ln J}{t}\frac{\partial }{\partial J}+\frac{\sigma ^2}{2}X^2\frac{\sigma ^2 V}{\partial X^2}+(r-q)X\frac{\partial V}{\partial X}-rV=0}
\\
&V(X,J,T)=(J_T-K)^{+}\nonumber
\end{eqnarray}
on the domain  $\lbrace 0\le X<\infty,0\le J<\infty,0\le t\le T\rbrace$.

\begin{lemma}\cite{J}\label{lemma}
The price of the geometric average Asian call option with fixed exercise price is given by
\begin{equation}\label{eq-22}
\displaystyle{V(X,J,t)=e^{-r(T-t)}\{ [J^tS^{T-t}]^{\frac{1}{T}}e^{\left(r_*+\frac{(\sigma ^*)^2}{2}\right)(T-t)}N(d_1^*)-KN(d_2^*) \}}
\end{equation}
Here

$$\displaystyle{d_1^*=\frac{\displaystyle{\frac{1}{T}\ln \frac{J^tX^{T-t}}{K^T}+[r^*+(\sigma ^*)^2](T-t)}}{\displaystyle{\sigma ^*\sqrt{T-t}}},d_2^*=d_1^*-\sigma ^*\sqrt{T-t}}$$
$$r^*=(r-q-\frac{\sigma ^2}{2})\frac{T-t}{2T},\sigma ^*=\frac{\sigma (T-t)}{\sqrt{3}T}$$
\end{lemma}

\indent

We now establish the differential equation model of the price of the discrete geometric average Asian option. let $0=T_1<\cdots <T_n=T$  be fixed monitoring times and $X_1,\cdots ,X_n$  be the price of the asset at time $T_1,\cdots ,T_n$  respectively. The price of discrete geometric average Asian option with fixed price with n monitoring times is denoted by $V_n(X,t)$ . By definition the expiry payoff is equal to $V_n(X_n,T_n)=\left( \sqrt[n]{X_1 \cdots X_n}-K\right)^{+}$. The price of the option on the time interval $T_i <t \le T_i+1$ is denoted by $V_n^i(X,t),i=1,\cdot ,n-1$ . Since in the time interval $T_{n-1}<t\le T_n$   the price of asset $X_1 ,\cdots ,X_{n-1}$  is known quantities, the expiry payoff can be written as

\begin{equation}\label{eq-23}
V_n^{n-1}(X,T_n)=\left(\sqrt[n]{X_1\cdots X_{n-1}}X^{\frac{1}{n}}-K \right)
\end{equation}
On the interval  $T_{n-1}<t\le T_n$, this can be seen as a normal vanilla option with expiry payoff \eqref{eq-23} and thus in this interval, the price of discrete geometrical average Asian option is the solution to the terminal value problem of Black-Scholes equation
$$LV_n^{n-1}=\frac{\partial V_n^{n-1}}{\partial t}+\frac{1}{2}\sigma ^2X^2\frac{\partial ^2V_n^{n-1}}{\partial X^2}+(r-q)X\frac{\partial V_n^{n-1}}{\partial X}-rV_n^{n-1}=0,(X>0,T_{n-1}<t<T_n),$$

\begin{equation}\label{eq-24}
V_n^{n-1}(X,T_n)=\left( \sqrt[n]{X_1 \cdots X_{n-1}}X^{\frac{1}{n}}-K\right)^{+}
\end{equation}

This price $V_n^{n-1}(X,t)$  depends $X,t$  as well as $X_1,\cdots ,X_{n-1}$ and so $V_n^{n-1}$  can be written as $V_n^{n-1}(X,t;X_1,\cdots X_{n-1})$.  Since $X=X_{n-1}$  at time $T_{n-1}$, especially  the option price is $V_n^{n-1}(X,T_n-1;X_1,\cdots ,X_{n-2},X)$

On the interval $T_{n-2}<t\le T_{n-1}$, the prices $X_1,\cdots ,X_{n-2}$, at the monitoring time $T_1,\cdots ,T_{n-2}$  at $T_{n-2}<t\le T_{n-1}$  is the known quantities and thus at this interval the option becomes a vanilla option (the solution of BS pde) whose payoff is $V_n^{n-1}(X,T_n-1;X_1,\cdots ,X_{n-2},X)$ . The price depends on  $X,t$ as well as $X_1,\cdots ,X_{n-2}$ . By repeating this process we get
\begin{equation}\label{eq-25}
LV_n^i=0~~~~~~~~~~~~(X>0,T_i<t<T_{i+1})
\end{equation}
$$V_n^i(X,T_i)=V_n^{i+1}(X,T_i;X_1\cdots ,X_i,X),i=\bar{1,n-2}$$
The model of the price of discrete geometric average Asian option with fixed exercise price with n monitoring times is \eqref{eq-24},\eqref{eq-25}.
\begin{theorem} \label{the-5}
The price of discrete geometric average Asian options with fixed exercise strike price with  n monitoring times is expressed by
$$V_n^{n-k}(X,t)=\sqrt[n]{X_1\cdots X_{n-k}}X^{\frac{k}{n}}e^{\theta _k(t)}N(d_2^{n-k}(t)+\frac{1}{n}\Delta _k(t))-Ke^{-r(T_n-t)}N(d_2^{n-k}(t))$$

\begin{eqnarray}\label{eq-26}
T_{n-k}\le t<T_{n-k+1},(k=\bar{1,n-1})
\end{eqnarray}
Here

$$\Delta _k(t)=\sigma \sqrt{k^2(T_{n-k+1}-t)+\sum_{i=1}^{k-1}i^2(T_{n-i+1}-T_i))}$$
$$d_2^{n-k}(X,t)=\Delta _{n-k}^{-1}\{\ln\frac{X_1\cdots X_{n-k}X^k}{K^n}+\left(r-q-\frac{\sigma ^2}{2}\right)[k(T_{n-k+1}-t)+\sum_{i=1}^{k-1}i(T_{n-i+1}-T_{n-i})]\}$$
\begin{eqnarray*}
\theta _k(t) &=& \mu (\frac{k}{n})(T_{n-k+1}-t)+\sum_{i=1}^{k-1}\mu \left(\frac{i}{n}(T_{n-i+1}-T_i)\right)
\\
&=& \mu(\frac{1}{n})(T_n-T_{n-1})+\mu (\frac{2}{n})(T_{n-1}-T_{n-2})+\cdots +\mu (\frac{k}{n})(T_{n-k+1}-t)
\end{eqnarray*}
Especially at time  $T_1=0$ we have
$$V_n^1(X,0)=Xe^{\theta _{n-1}(T_1)}N(d_2+\frac{1}{n}\Sigma)-Ke^{-r(T_n-t)}N(d_2)$$
Here

$$\Sigma =\Delta_{n-1}(0)\sigma \sqrt{\sum_{i=1}^{n-1}(i^2(T_{n-k+1}-T_{n-i}))}$$
$$d_2=d_2^1(X,T_1)=\Sigma ^{-1}\left[\ln \frac{X^n}{K^n}+(r-q-\frac{\sigma ^2}{2})\sum_{i=1}^{n-i}{i(T_{n-i+1}-T_{n-i})}\right]$$
$$\theta _{n-1}(T_1)=\sum_{i=1}^{n-1}{\mu (\frac{i}{n})(T_{n-i+1-T_{n-i}})}$$

\end{theorem}

\begin{proof}
 First we find a solution of \eqref{eq-24}.The expiry payoff of \eqref{eq-24} can be written as

\begin{eqnarray*}
V_n^{n-1}(X,T_n) &=& (\sqrt[n]{X_1\cdots X_{n-1}}X^{\frac{1}{n}}-K)^{+}
\\
&=& (\sqrt[n]{X_1\cdots X_{n-1}}X^{\frac{1}{n}}-K)\cdot 1(\sqrt[n]{X_1\cdots X_{n-1}}X^{\frac{1}{n}}>K)
\\
&=& (\sqrt[n]{X_1\cdots X_{n-1}}X^{\frac{1}{n}}-K)\cdot 1\left(X>\frac{K^n}{X_1\cdots X_{n-1}}\right).
\end{eqnarray*}
This is expressed by the combination of the prices of power binary options and thus from Theorem \ref{theorem 2} for $T_{n-1}<t<T_n$  we have
\begin{eqnarray}
V_n^{n-1}(X,t;T_n) &=& \sqrt[n]{X_1\cdots X_{n-1}} \left(M^{\frac{1}{n}} \right)_{\frac{K^n}{X_1\cdots X_{n-1}}}^{+}(X,t;T_n)-K(M^0)_{\frac{K^n}{X_1\cdots X_{n-1}}}^{+}(T,t;T_n)\nonumber
\\
&=&  \sqrt[n]{X_1\cdots X_{n-1}}e^{\mu (\frac{1}{n})(T_n-t)}N(d_1^n(X,t))-Ke^{-r(T_n-t)}N(d_2^n(X,t))
\end{eqnarray}
Here
\begin{equation}
d_2^{n-1}(X,t)=\left[\ln \frac{X_1\cdots X_{n-1}}{K^n}+\left(r-q-\frac{\sigma ^2}{2}\right)(T_n-t)\right](\sigma \sqrt{T_n-t})^{-1}\nonumber
\end{equation}
\begin{eqnarray}
d_1^{n-1}(X,t) &=& \left[\ln \frac{X_1\cdots X_{n-1}}{K^n}+\left(r-q-\frac{\sigma ^2}{2}+\frac{\sigma ^2}{n}\right)(T_n-t)\right](\sigma \sqrt{T_n-t})^{-1}\nonumber
\\
&=& d_2^{n-1}(X,t)+\frac{1}{n}\sigma \sqrt{T_n-t}=d_2^{n-1}(X,t)+\frac{1}{n}\Delta _1(t)\nonumber
\end{eqnarray}
Thus \eqref{eq-26} has been proved for $T_{n-1}<t<T_n$

Now we assume that \eqref{eq-26} holds on the interval $T_{n-k}<t<T_{n-k+1}$   and find the price formula on the interval $T_{n-k-1}<t<T_{n-k}$. On the interval $T_{n-k-1}<t<T_{n-k}$,  $X_1,\cdots,X_{n-k-1}$  is known quantities and especially $X=X_{n-k}$  at time $T_{n-k}$ . Thus at time $T_{n-k}$, we have

$V_n^{n-k}(X,T_{n-k};T_{n-k+1},T_{n-1},T_n)=$
\begin{eqnarray*}
&=&  \sqrt[n]{X_1\cdots X_{n-k-1}}X^{(k+1)/n}e^{\theta _k(n-k)}N(d_1^{n-k}(X,T_{n-k}))-ke^{-r(T_n-t)}N(d_2^{n-k}(X,T_{n-k}))
\\
&=& f(X)-g(X)
\end{eqnarray*}
Here
\begin{eqnarray*}
d_2^{n-k}(X,T_{n-k}) &=& \left\{ \ln \frac{X_1\cdots X_{n-k-1}X^{k+1}}{K^n}+(r-q-\frac{\sigma ^2}{2})\left[ \sum_{i=1}^k{i(T_{n-k+1}-T_{n-i})}\right]\right\}
\\
&\cdot& \left(\sigma \sqrt {\sum_{i=1}^{k}{i^2(T_{n-i+1}-T_{n-i})}}\right)
\\
&=& \delta \left(\frac{X_1\cdots X_{n-k-1}X^{k+1}}{K^n},\sum_{i=1}^k{i(T_{n-k+1}-T_{n-i})},0,\sum_{i=1}^{k}{i^2(T_{n-i+1}-T_{n-i})}\right)
\end{eqnarray*}
\begin{eqnarray*}
d_1^{n-k}(X,T_{n-k})&=& d_2^{n-k}(X,T_{n-k})+\frac{1}{n}\Delta _k(T_{n-k}) 
\\
&=& \big\{ \ln \frac{X_1\cdots X_{n-k-1}X^{k+1}}{K^n}+(r-q-\frac{\sigma ^2}{2})\left[ \sum_{i=1}^k{i(T_{n-k+1}-T_{n-i})}\right]
\\
&+& \frac{\sigma ^2}{n}\sum_{i=1}^k{i^2(T_{n-i+1}-T_{n-i})} \big\} \cdot \left(\sigma \sqrt {\sum_{i=1}^{k}{i^2(T_{n-i+1}-T_{n-i})}}\right)
\\
&=& \delta \left(\frac{X_1\cdots X_{n-k-1}X^{k+1}}{K^n},\sum_{i=1}^k{i(T_{n-k+1}-T_{n-i})},\frac{\frac{1}{n}\sum_{i=1}^{k}{i^2(T_{n-i+1}-T_{n-i})}}{\sum_{i=1}^k{i(T_{n-k+1}-T_{n-i})}},\sum_{i=1}^{k}{i^2(T_{n-i+1}-T_{n-i})}\right)
\end{eqnarray*}

where $\delta$ is defined by \eqref{eq-11}. And from \eqref{eq-25},  $V_n^{n-k-1}(X,t)$ is the solution of \eqref{eq-1} with expiry payoff $f(X)-g(X)$. Now in order to find the solution  $V_f(X)$ of \eqref{eq-1} with the expiry payoff $f(X)$ ,we let
$$\tau _1=\sum_{i=1}^k{i(T_{n-i+1}-T_{n-i})},\tau _2=T_{n-k}-t,\tau_1'=\sum_{i=1}^k{i^2(T_{n-i+1}-T_{n-i})}$$
$$\alpha =\frac{\frac{1}{n}\sum_{i=1}^{k}{i^2(T_{n-i+1}-T_{n-i})}}{\sum_{i=1}^k{i(T_{n-k+1}-T_{n-i})}},\beta =\frac{k+1}{n},i=k+1$$

and apply Theorem \ref{theorem 3}. Thus we have
$$V_f(X,t)=\sqrt[n]{X_1\cdots X_{n-k-1}}X^{(k+1)/n}e^{\theta _{k+1}(t)}N(d_1^{n-k-1})$$
Here
\begin{eqnarray*}
d_1^{n-k-1}(X,t)&=& \bigg\{ \ln \frac{X_1\cdots X_{n-k-1}X^{k+1}}{K^n}+(r-q-\frac{\sigma ^2}{2})\left[(k+1)(T_{n-k}-t)+\sum_{i=1}^k{i(T_{n-i+1}-T_{n-i})}\right]+
\\
&+& \frac{\sigma ^2}{n}\left[(k+1)^2(T_{n-k}-t)+\sum_{i=1}^{k}{i^2(T_{n-i+1}-T_{n-i})}\right]\bigg\}
\\
&\cdot & \left(\sigma \sqrt{(k+1)^2(T_{n-k}-t)+\sum_{i=1}^k}{i^2(T_{n-i+1}-T_{n-i})}\right)^{-1}
\end{eqnarray*}
$\theta _{k+1}(t)=\mu (\frac{k+1}{n})(T_{n-k+1}-t)+\sum_{i=1}^k{\mu (\frac{i}{n})(T_{n-i+1}-T_{n-i})}$

In order to find the solution $V_g(X,t)$ of \eqref{eq-1} with expiry payoff $g(X)$ ,we let
$$\tau _1=\sum_{i=1}^k{i(T_{n-i+1}-T_{n-i})},\tau _2=T_{n-k}-t,\tau _1'=\sum_{i=1}^k{i^2(T_{n-i+1}-T_{n-i})}$$
$$\alpha =0,\beta =0,i=k+1$$
and apply Theorem \ref{theorem 3}. Thus we have
$$V_g(X,t)=Ke^{-r(T_n-t)}N(d_2^{n-k-1})$$
Here
\begin{eqnarray*}
d_2^{n-k-1}(X,t)&=& \bigg\{ \ln\frac{X_1\cdots X_{n-k-1}X^{k+1}}{K^n}+(r-q-\frac{\sigma ^2}{2})\left[(k+1)(T_{n-k}-t)+\sum_{i=1}^k{i(T_{n-i+1}-T_{n-i})}\right]\bigg\}
\\
&\cdot & \left(\sigma \sqrt{(k+1)^2(T_{n-k}-t)+\sum_{i=1}^{k}{i^2(T_{n-i+1}-T_{n-i})}}\right)^{-1}
\end{eqnarray*}
So we have
\begin{eqnarray*}
V_n^{n-k-1}(X,t) &=& V_f(X,t)-V_g(X,t)
\\
&=& \sqrt[n]{X_1\cdots X_{n-k-1}}X^{(k+1)/n}e^{\theta _{k+1}(t)}N(d_1^{n-k-1})-Ke^{-r(T_n-t)}N(d_2^{n-k-1})
\end{eqnarray*}
Thus we have proved \eqref{eq-26}.(QED)
\end{proof}

And now let's find the limit of the prices of discrete geometric averageAsian options as discrete monitoring interval goes to zero. Let $V_n(X,t)$   be the price of discrete geometric average Asian options with fixed exercise price and n monitoring times. Then we have 
$$\forall t\in [0,T),\exists k\in \{1,\cdots ,n-1\},T_{n-k}\le t<T_{n-k+1}:V_n(X,t)=V_n^{n-k}(X,t)$$
And let $V(X,J,t)$  be the price of continuous geometric average Asian call option with fixed exercise price and $J(t)=e^{\frac{1}{t}\int_{0}^{t}{\ln X_\tau d\tau}}$.Then we have the following covergence theorem.

%%%%%%%%%%%%
%%%%%%%%%%%%             theorem 3.3
\begin{theorem}\label{the-6}
As the maximum length of the subintervals of the partition goes to zero and let n go to infinity, then the price of \eqref{eq-26} of discrete geometric average Asian call options converges to the price of \eqref{eq-22} of continuous geometric average Asian options. That is, we have
$$\lim_{n \rightarrow \infty}{V_n{(X,~t)}}=V(X,J,t),\forall t\in [0,T)$$
\end{theorem}
\begin{proof}
For simplicity of discussion, fix time $t(0<t<T)$ and assume that $\{0=T_1,T_2,\cdots ,T_n=T\}$  is the partition of $[0,T]$ wiht (n-1) subintervals with same length. (In the case with any partition, the convergence can be proved in the same way.) Then $\exists k,T_{n-k}\le t<T_{n-k+1}$ and we have
\begin{eqnarray*}
V_n(X,t)&=& V_n^{n-k}(X,t;T_{n-k+1},\cdots ,T_{n-1},T_n)=
\\
&=& \sqrt[n]{X_1\cdots X_{n-k}}X^{k/n}e^{\theta _k}N(d_1^{n-k})+Ke^{-r(T_n-t)}N(d_2^{n-k})
\end{eqnarray*}
Here for convenience   $d_1^{n-k}(X,t)$ and $d_2^{n-k}(X,t)$ can be rewritten as following.

\begin{eqnarray*}
d_1^{n-k}(X,t) &=& \bigg\{ \ln\frac{X_1\cdots X_{n-k}X^k}{K^n}+(r-q-\frac{\sigma ^2}{2})\bigg[k(T_{n-k+1}-t)+\sum_{i=1}^{k-1}{i(T_{n-i+1}-T_{n-i})}\bigg]+
\\
&+& \frac{\sigma ^2}{n}\left[k^2(T_{n-k+1}-t)+\sum_{i=1}^{k-1}{i^2(T_{n-i+1-T_{n-i}})}\right]\bigg\}\cdot
\\
&\cdot &\left(\sigma \sqrt{k^2(T_n-k+1)-t}+\sum_{i=1}^{k-1}{i^2(T_{n-i+1}-T_{n-i})}\right)^{-1}
\\
&=& \bigg\{\frac{n}{k} \ln\frac{\sqrt[n]{X_1\cdots X_{n-k}}X^{k/n}}{K^n}+(r-q-\frac{\sigma ^2}{2})\bigg[(T_{n-k+1}-t)+\frac{1}{k}\sum_{i=1}^{k-1}{i(T_{n-i+1}-T_{n-i})}\bigg]+
\\
&+& \frac{k\sigma ^2}{n}\left[(T_{n-k+1}-t)+\frac{1}{k^2}\sum_{i=1}^{k-1}{i^2(T_{n-i+1-T_{n-i}})}\right]\bigg\}\cdot
\\
&\cdot &\left(\sigma \sqrt{(T_n-k+1)-t}+\frac{1}{k^2}\sum_{i=1}^{k-1}{i^2(T_{n-i+1}-T_{n-i})}\right)^{-1}
\end{eqnarray*}
\begin{eqnarray*}
d_2^{n-k}(X,t) &=& \bigg\{\frac{n}{k} \ln\frac{\sqrt[n]{X_1\cdots X_{n-k}}X^{k/n}}{K^n}+(r-q-\frac{\sigma ^2}{2})\bigg[(T_{n-k+1}-t)+\frac{1}{k}\sum_{i=1}^{k-1}{i(T_{n-i+1}-T_{n-i})}\bigg]\bigg\}\cdot
\\
&\cdot &\left(\sigma \sqrt{(T_n-k+1)-t}+\frac{1}{k^2}\sum_{i=1}^{k-1}{i^2(T_{n-i+1}-T_{n-i})}\right)^{-1}
\end{eqnarray*}
Then we have
%%%%%%%%
$$T_{n-k}\le t<T_{n-k+1}\leftrightarrow \frac{n-k-1}{n-1}T\le <t<\frac{n-k}{n-1}T\leftrightarrow \frac{k-1}{n-1}<\frac{T-t}{T}\le \frac{k}{n-1}$$
So we have
\begin{equation}\label{eq-28}
\lim_{n\rightarrow \infty}{\frac{k}{n}}=\frac{T-t}{T}
\end{equation}

Using \eqref{eq-28} we have

\begin{eqnarray*}
\ln\left({J(t)}^{\frac{t}{T}}\right) &=& \frac{t}{T}\cdot \frac{1}{t}\int_0^t{\ln X_{\tau}d\tau }=\frac{1}{T}\lim_{n\rightarrow \infty}{\sum_{i=1}^{n-k}{\left(\frac{T}{n-1}\ln X_i\right)}}=
\\
&=& \lim_{n\rightarrow \infty}{\frac{1}{n}\sum_{i=1}^{n-k}\ln X_i}=\lim_{n \rightarrow \infty}{\ln \sqrt[n]{X_1\cdots X_{n-k}}}=
\\
&=& \ln \lim_{n \rightarrow \infty}\sqrt[n]{X_1\cdots X_{n-k}}
\end{eqnarray*}
So we have
\begin{equation}\label{eq-29}
{J(t)}^{\frac{t}{T}}=\lim_{n \rightarrow \infty}\sqrt[n]{X_1\cdots X_{n-k}}
\end{equation}

\begin{eqnarray*}
(T_{n-k+1}-t)+\frac{1}{k^2}\sum_{i=1}^{k-1}{i^2(T_{n-i+1}-T_{n-i})} &=& (T_{n-k+1}-t)+\frac{T}{n-1}\frac{1}{k^2}\sum_{i=1}^{k-1}{i^2}
\\
&=& (T_{n-k+1}-t)+\frac{T}{n-1}\frac{k(k-1)(2k-1)}{6k^2}=
\\
&=& (T_{n-k+1}-t)+\frac{k-1}{n-1}\left(\frac{1}{3}-\frac{1}{6k}\right)T
\end{eqnarray*}
Using \eqref{eq-28} and $\lim_{n \rightarrow \infty}{(T_{n-k+1}-t)}=0$, we have
\begin{equation}\label{eq-30}
\lim_{n \rightarrow \infty}{(T_{n-k+1}-t)+\frac{1}{k^2}\sum_{i=1}^{k-1}{i^2(T_{n-i+1}-T_{n-i})}}=\frac{1}{3}(T-t)
\end{equation}

Using $\frac{1}{k}\ln \frac{X_1\cdots X_{n-k}X^k}{K^n}=\frac{n}{k}\ln \frac{\sqrt[n]{X_1\cdots X_{n-k}}X^{\frac{k}{n}}}{K}$ and \eqref{eq-28} and \eqref{eq-29},we have

\begin{equation}\label{eq-31}
\lim_{n \rightarrow \infty}{\frac{n}{k}\ln \frac{\sqrt[n]{X_1\cdots X_{n-k}}X^{k/n}}{K}}=\frac{T}{T-t}\ln{\frac{{J(t)}^{\frac{t}{T}}X^{\frac{T-t}{T}}}{K}}
\end{equation}

Using $\frac{1}{k}\sum_{i=1}^{k-1}{i(T_{n-i+1}-T_{n-i})}+(T_{n-k+1}-t)= \frac{T}{n-1}\frac{k-1}{2}+(T_{n-k+1}-t)$ and \eqref{eq-28} we have

\begin{equation}\label{eq-32}
\lim_{n \rightarrow \infty}{\frac{1}{k}\sum_{i=1}^{k-1}{i(T_{n-i+1}-T_{n-i})}+(T_{n-k+1}-t)}=\frac{T}{2}\frac{T-t}{T}=\frac{T-t}{2}
\end{equation}
Now rewrite as
\begin{eqnarray*}
d_1^{n-k}(X,t) &=& \bigg\{ \ln\frac{X_1\cdots X_{n-k}X^k}{K^n}+(r-q-\frac{\sigma ^2}{2})\bigg[k(T_{n-k+1}-t)+\sum_{i=1}^{k-1}{i(T_{n-i+1}-T_{n-i})}\bigg]+
\\
&+& \frac{\sigma ^2}{n}\left[k^2(T_{n-k+1}-t)+\sum_{i=1}^{k-1}{i^2(T_{n-i+1-T_{n-i}})}\right]\bigg\}\cdot
\\
&\cdot &\left(\sigma \sqrt{k^2(T_n-k+1)-t}+\sum_{i=1}^{k-1}{i^2(T_{n-i+1}-T_{n-i})}\right)^{-1}
\\
&=& \bigg\{\frac{n}{k} \ln\frac{\sqrt[n]{X_1\cdots X_{n-k}}X^{k/n}}{K^n}+(r-q-\frac{\sigma ^2}{2})\bigg[(T_{n-k+1}-t)+\frac{1}{k}\sum_{i=1}^{k-1}{i(T_{n-i+1}-T_{n-i})}\bigg]+
\\
&+& \frac{k\sigma ^2}{n}\left[(T_{n-k+1}-t)+\frac{1}{k^2}\sum_{i=1}^{k-1}{i^2(T_{n-i+1-T_{n-i}})}\right]\bigg\}\cdot
\\
&\cdot &\left(\sigma \sqrt{(T_n-k+1)-t}+\frac{1}{k^2}\sum_{i=1}^{k-1}{i^2(T_{n-i+1}-T_{n-i})}\right)^{-1}
\end{eqnarray*}
From \eqref{eq-30},\eqref{eq-31},\eqref{eq-32} we have
\begin{eqnarray}\label{eq-33}
\lim_{n \rightarrow \infty}{d_1^{n-k}(X,t)}&=& \frac{\frac{T}{T-t}\ln \frac{{J(t)}^{\frac{t}{T}}X^{\frac{T-t}{T}}}{K}+(r-q-\frac{\sigma ^2}{2})\frac{T-t}{2}+\frac{T-t}{T}\frac{T-t}{T}\frac{T-t}{3}\sigma ^2}{\sigma \sqrt{\frac{T-t}{3}}}=\nonumber
\\
&=& \frac{\frac{1}{T}\ln \frac{{J(t)}^{t}}X^{T-t}{K^T}+\left((r-q-\frac{\sigma ^2}{2})\frac{T-t}{2T}+\frac{1}{3}\left(\frac{T-t}{T}\right)^2\sigma ^2\right)(T-t)}{\sigma \frac{T-t}{T}\frac{\sqrt{T-t}}{\sqrt{3}}}=\nonumber
\\
&=& \frac{\frac{1}{T}\ln \frac{{J(t)}^{t}}X^{T-t}{K^T}+[r^*+(\sigma ^*)^2](T-t)}{\sigma ^*\sqrt{T-t}}=d_1^*
\end{eqnarray}
Similarly rewrite as
\begin{eqnarray*}
d_2^{n-k}(X,t) &=& \bigg\{\frac{n}{k} \ln\frac{\sqrt[n]{X_1\cdots X_{n-k}}X^{k/n}}{K^n}+(r-q-\frac{\sigma ^2}{2})\bigg[(T_{n-k+1}-t)+\frac{1}{k}\sum_{i=1}^{k-1}{i(T_{n-i+1}-T_{n-i})}\bigg]\bigg\}\cdot
\\
&\cdot &\left(\sigma \sqrt{(T_n-k+1)-t}+\frac{1}{k^2}\sum_{i=1}^{k-1}{i^2(T_{n-i+1}-T_{n-i})}\right)^{-1}
\end{eqnarray*}

From \eqref{eq-30},\eqref{eq-31},\eqref{eq-32} we have
\begin{eqnarray}\label{eq-34}
\lim_{n \rightarrow \infty}{d_2^{n-k}(X,t)}&=& \frac{\frac{T}{T-t}\ln \frac{{J(t)}^{\frac{t}{T}}X^{\frac{T-t}{T}}}{K}+(r-q-\frac{\sigma ^2}{2})\frac{T-t}{2}}{\sigma \sqrt{\frac{T-t}{3}}}=\nonumber
\\
&=& \frac{\frac{1}{T}\ln \frac{{J(t)}^{t}}X^{T-t}{K^T}+\left((r-q-\frac{\sigma ^2}{2})\frac{T-t}{2T}\right)(T-t)}{\sigma \frac{T-t}{T}\frac{\sqrt{T-t}}{\sqrt{3}}}=\nonumber
\\
&=& \frac{\frac{1}{T}\ln \frac{{J(t)}^{t}}X^{T-t}{K^T}+r^*(T-t)}{\sigma ^*\sqrt{T-t}}=d_2^*
\end{eqnarray}
where
$$r^*=(r-q-\frac{\sigma ^2}{2})\frac{T-t}{2T},\sigma ^*=\frac{T-t}{\sqrt{3}T}$$
Consider
\begin{eqnarray*}
\theta _k(t)&=& \mu (\frac{1}{n})(T_n-T_{n-1})+\mu (\frac{2}{n})(T_{n-1}-T_{n-2})+\cdots +\mu (\frac{k}{n})(T_{n-k+1}-t)=
\\
&=& \left(\frac{T}{n-1}\sum_{i=1}{k-1}(\frac{i}{n}-1)+(\frac{k}{n}-1)(T_{n-k+1}-t)\right)r-\left(\frac{T}{n-1}\sum_{i=1}^{k-1}+\frac{k}{n}(T_{n-k+1}-t)\right)q-
\\
&-& \frac{\sigma ^2}{2}\left(\frac{T}{n-1}\sum_{i=1}{k-1}(\frac{i^2}{n^2}-\frac{i}{n})+(\frac{k^2}{n^2}-\frac{k}{n})(T_{n-k+1}-t)\right)=
\\
&=& \left(\frac{T}{n-1}\left(\frac{k(k-1)}{2n}-k\right)(\frac{k}{n}-1)(T_{n-k+1}-t)\right)r-\left(\frac{T}{n-1}\frac{k(k+1)}{2n}+\frac{k}{n}(T_{n-k+1}-t)\right)-
\\
&-& \frac{\sigma ^2}{2}\left(\frac{T}{n-1}\left(\frac{k(k-1)(2k-1)}{6n^2}+\frac{k(k-1)}{2n}\right)+\left(\frac{k^2}{n^2}-\frac{k}{n}\right)(T_{n-k+1}-t)\right)
\end{eqnarray*}
From \eqref{eq-30},\eqref{eq-31},\eqref{eq-32} we have

\begin{eqnarray}\label{eq-35}
\lim_{n\rightarrow \infty}{\theta _k(t)}&=& \left(\frac{(T-t)^2}{2T}-(T-t)\right)r-\frac{(T-t)^2}{2T}q-\left(\frac{(T-t)^3}{3T^2}+\frac{(T-t)^2}{2T}\right)\frac{\sigma ^2}{2}=\nonumber
\\
&=& (r-q-\frac{\sigma ^2}{2})\frac{(T-t)^2}{2T}+\frac{(T-t)^3}{3T^2}\frac{\sigma ^2}{2}-r(T-t)=\nonumber
\\
&=& (r^*+\frac{(\sigma ^*)^2}{2}(T-t)-r(T-t))
\end{eqnarray}
From \eqref{eq-33},\eqref{eq-34},\eqref{eq-35} we have the required result (QED).

\end{proof}

We can obtain the result of the geometric Asian options with floating price in the same way as in Theorem \ref{the-5} and Theorem \ref{the-6}.

Let $0=T_1<\cdots <T_n=T$  be the fixed monitring times and $X_1,\cdots ,X_n$  the asset prices at the monitoring times, respectively. Denote the price of discrete geometric Asian option with floating exercise price by $V_n(X,t)$ . From the definition, the expriy payoff is
$$V_n(X_n,T_n)=\left(X-\sqrt[n]{X_1\cdots X_n}\right)^+$$
Denote the option's price on the interval $T_i<t\le X_n$  by $V_n^i(X,t),i=1,\cdots n-1$ . Then on this interval  $T_{n-1}<t<\le {T_n}$, the assets prices  $X_1,\cdots ,X_{n-1}$ are  known quantities, we can rewrite as.
\begin{equation}\label{eq-36}
V_n^{n-1}(X,T_n)=\left(X-\sqrt[n]{X_1\cdots X_{n-1}}X^{1/n}\right)^+
\end{equation}
Thus in this interval the discrete geometric Asian options with floating exercise price is a vanilla option with the expiry payoff (36) and the pricing model is given by
$$LV_n^{n-1}=0~~~~~~~~~~~~(X>0,T_{n-1}<t<T_n)$$
\begin{equation}\label{eq-37}
V_n^{n-1}(X,T_n)=\left(X-\sqrt[n]{X_1\cdots X_{n-1}}X^{1/n}\right)^+
\end{equation}

 $V_n^{n-1}(X,t)$ depends on $X,t $ and  $X_1,\cdots ,X_{n-1}$ and so $V_n^{n-1}(X,t)$  can be rewritten by $V_n^{n-1}(X,t;X_1,\cdots,X_{n-1})$ . Especially $X=X_{n-1}$  at time $T_{n-1}$ , the price $V_n^{n-1}(X,T_{n-1}$ can be rewritten as$V_n^{n-1}(X,t;X_1,\cdots,X_{n-2},X)$. On the interval $T_{n-2}<t\le T_{n-1}$  the asset prices $X_1,\cdots ,X_{n-2}$  at the monitoring times $T_1 ,\cdots ,T_{n-2}$ are known quantities and thus in this interval the option is a vanilla option with the expiry payoff  $V_n^{n-1}(X,t;X_1,\cdots,X_{n-2},X)$. Again $V_n^{n-2}(X,t)$  depends on  $X,t$ as well as $X_1,\cdots ,X_{n-2}$ . Repeating these process, we get

$$LV_n^{i}=0~~~~~~~~~~~~(X>0,T_{i}<t<T_{i+1})$$
\begin{equation}\label{eq-38}
V_n^{i}(X,T_n)=V_n^{i+1}(X,T_i;X_1\cdots X_i,X),i=\bar{1,n-2}.
\end{equation}

The pricing model of discrete geometric Asian options with floating exercise price and n monitoring times is \eqref{eq-37} and \eqref{eq-38}.

%%%%%%%%%%%%%%%     theorem 3.4
\begin{theorem}\label{theorem-8}
The price of discrete geometric Asian options with floating price and n monitoring times (the solution to\eqref{eq-37},\eqref{eq-38}) is given by
\begin{equation}\label{eq-39}
V_n^{n-k}(X,t;T_{n-k+1},\cdots ,T_{n-1},T_n)=Xe^{-q(T_n-t)}N(d_2^{n-k})-\sqrt[n]{X_1\cdots X_{n-k}}X^{k/n}e^{\theta _k(t)}N(d_1^{n-k})
\end{equation}
$$T_{n-k}\le t<T_{n-k+1},(k=\bar{1,n-1})$$
Here
\begin{eqnarray*}
d_1^{n-k}&=& \bigg\{\ln \frac{X^{\frac{n-k}{n-1}}}{\sqrt[n-1]{X_1\cdots X_{n-k}}}+(r-q-\frac{\sigma ^2}{2})\left[\frac{n-k}{n-1}(T_{n-k+1}-t)+\sum_{j=1}^{k-1}{\frac{n-j}{n-1}(T_n-j+1-T_{n-j})}\right]+
\\
&+& \sigma ^2\left[\frac{k(n-k)}{n(n-1)}(T_{n-k+1}-t)+\sum_{j=1}^{k-1}{\frac{j(n-j)}{n(n-1)}(T_{n-j+1}-T_{n-j})}\right]\bigg\}\cdot
\\
&\cdot & \left(\sigma \sqrt{\left(\frac{n-k}{n-1}\right)^2(T_{n-k+1}-t)+\sum_{j=1}^{k-1}{\left(\frac{n-j}{n-1}\right)^2(T_{n-j+1}-T_{n-j})}}\right)^{-1}
\end{eqnarray*}

\begin{eqnarray*}
d_2^{n-k}&=& \bigg\{\ln \frac{X^{\frac{n-k}{n-1}}}{\sqrt[n-1]{X_1\cdots X_{n-k}}}+(r-q-\frac{\sigma ^2}{2})\left[\frac{n-k}{n-1}(T_{n-k+1}-t)+\sum_{j=1}^{k-1}{\frac{n-j}{n-1}(T_n-j+1-T_{n-j})}\right]\cdot
\\
&\cdot & \left(\sigma \sqrt{\left(\frac{n-k}{n-1}\right)^2(T_{n-k+1}-t)+\sum_{j=1}^{k-1}{\left(\frac{n-j}{n-1}\right)^2(T_{n-j+1}-T_{n-j})}}\right)^{-1}
\end{eqnarray*}

$$\theta _k(t)=\mu (\frac{k}{n})(T_{n-k+1}-t)+\sum_{j=1}^{k-1}{\mu (\frac{j}{n})(T_{n-j+1}-T_{n-j})}$$
$$\mu (\beta)=(\beta -1)r-\beta q+\frac{\sigma ^2}{2}(\beta ^2-\beta)$$
\end{theorem}
The proof is omitted as it is similar with the proof of Theorem \ref{the-5}.

%%%%%%%%%%%%%%%     theorem 3.5
\begin{theorem}
In the case with floating exercise price, if we increase the number $n$ of monitoring times, the price of discrete geometric Asian options  converges to the price of $V(X,J,t)$ continuous geometric Asian options with floating exercise price. That is, we have
$$\lim_{n\rightarrow \infty}{V_n(X,t)}=V(X,J,t)$$
Here $V(X,J,t)$  is given by
$$V(X,J,t)=e^{-q(T-t)}XN(d_2^*)-J^{\frac{t}{T}}X^{\frac{T-t}{T}}e^{\theta ^*}N(d_1^*)$$
where

$$d_1^*=\frac{\sqrt{3}}{\sigma \sqrt{T^3-t^3}}\left[t\ln{\frac{X}{J}}+(r-q+\frac{\sigma ^2}{2})\frac{T^2-t^2}{2}-\sigma ^3\frac{T^3-t^3}{3T}\right]$$
$$d_2^*=\frac{\sqrt{3}}{\sigma \sqrt{T^3-t^3}}\left[t\ln{\frac{X}{J}}+(r-q+\frac{\sigma ^2}{2})\frac{T^2-t^2}{2}\right]$$
$$\theta ^*=-q(T-t)-(r-q+\frac{\sigma ^2}{2})\frac{T^2-t^2}{2T}+\frac{\sigma ^2(T^3-t^3)}{6T^2}$$
\end{theorem}
The proof is omitted as it is similar with the proof of Theorem \ref{the-6}.

\textbf{Remark}: Remark: This formula is slightly different from that of \cite{J}. But it is equal to the formula of \cite{P2} in the case that q=0.

%%%%%%%%%%%%Section 4%%%%%%%%%%

\section{High order power binary option}

\begin{definition}
~\textbf{A second order  $\alpha$-power binary option} is defined as the binary contract with expiry date $T_0$  underlying on an $\alpha$-power binary option with expiry date $T_1$. In other words, a second order $\alpha$-power binary option's price is the solution of \eqref{eq-1} with time $T_0$ payoff of the following form 
\end{definition}
$$V(x,T_0)=(M^\alpha)_{\xi_1}^{s_1}(x,T_0;T_1)1({s_0}x>{s_0}{\xi _0})$$

Here $(M^\alpha)_{\xi_1}^{s_1}(x,T_0;T_1)$ is the price at time $T_0$ of the $\alpha$-power binary option with expiry date $T_1$.
The price of the second  $\alpha$-power binary option is denoted by $(M^\alpha)^{s_0 s_1}_{\xi _0 \xi _1}(x,t;T_0,T_1)$  and this is called the second order $\alpha$-power binary option with expiry dates  $T_0,T_1$.
Correspondingly, $\alpha$- power binary option is called the first order  $\alpha$-power binary option.

Since
$$(M^\alpha)^{s_0s_1}_{\xi _0\xi _1}(x,T_0;T_0,T_1)=(M^\alpha)^{s_1}_{\xi _1}(x,T_0;T_1)1(s_0x>s_0\xi _0)$$
we have
$$(M^\alpha)^{+s_1}_{\xi _0\xi _1}(x,T_0;T_0,T_1)+(M^\alpha)^{-s_1}_{\xi _0\xi _1}(x,T_0;T_0,T_1)=(M^\alpha)^{s_1}_{\xi _1}(x,T_0;T_1)$$

Thus we have a parity relation between the price of the first order  $\alpha$-power binary option and the prices of the corresponding second order $\alpha$-power binary options.
$$(M^\alpha)^{+s_1}_{\xi _0\xi _1}(x,t)+(M^\alpha)^{-s_1}_{\xi _0\xi _1}(x,t)=(M^\alpha)^{s_1}_{\xi _1}(x,t),t<T_0$$

\begin{theorem}\label{theorem 5}
The price of the second order power binary option is given by
\begin{equation}\label{eq-40}
(M^\alpha)^{s_0s_1}_{\xi _0\xi _1}(x,t;T_0,T_1)=e^{\mu (T_1-t)}x^{\alpha}N_2(s_0d_0,s_1d_1;s_0s_1\rho)
\end{equation}
Here
$$d_i=\left( \ln\frac{x}{\xi _i}+\left(r-q-\frac{\sigma ^2}{2}+\alpha \sigma ^2\right)(T_i-t)\right)(\sigma \sqrt{T_i-t})^{-1},\rho =\sqrt{\frac{T_0-t}{T_1-t}}$$
$$N_2(d_0,d_1;\rho)=\frac{1}{2\pi \sqrt{1-\rho ^2}}\int_{-\infty}^{d_0}\int_{-\infty}^{d_1}e^{-\frac{y_0^2-2\rho y_0y_1+y_1^2}{2(1-\rho ^2)}}dy_0dy_1$$
\end{theorem}
\begin{proof}
 By the definition and Preposition 1 of \cite{OK} we have

\begin{multline}
(M^\alpha)^{s_0 s_1}_{\xi _0 \xi _1}(x,t)=\nonumber
\\
= \frac{e^{-r(T_0-t)}}{\sigma \sqrt{2\pi (T_0-t)}}\int_{-\infty}^{\infty}\frac{1}{z_0}e^{-\frac{1}{2\sigma ^2(T_0-t)}\left(\ln\frac{x}{z_0}+\left(r-q-\frac{\sigma ^2}{2}\right)(T_0-t)\right)^2}1(s_0z_0>s_0\xi _0)(M^\alpha )^{s_1}_{\xi_1}(z_0,T_0;T_1)dz_0=
\\
= e^{-r(T_0-t)}\int_{-\infty}^{+\infty}\frac{e^{-\frac{1}{2\sigma ^2(T_0-t)}\left(\ln\frac{x}{z_0}+\left(r-q-\frac{\sigma ^2}{2}\right)(T_0-t)\right)^2}}{z_0 \sigma \sqrt{2\pi (T_0-t)}}1(s_0z_0>s_0\xi _0)
\\
~~~~~~\left( e^{\mu(T_1-T_0)}z_0^\alpha\int_{-\infty}^{+\infty}\frac{e^{-\frac{1}{2\sigma ^2(T_1-T_0)}\left(\ln\frac{x}{z_1}+\left(r-q-\frac{\sigma ^2}{2}+\alpha \sigma^2\right)(T_1-T_0)\right)^2}}{z_1 \sigma \sqrt{2\pi (T_1-T_0)}}1(s_1z_1>s_1\xi _1)dz_1\right)dz_0
\\
=x^\alpha e^{\mu (T_1-t)}\int_{-\infty}^{+\infty}\int_{-\infty}^{+\infty}\left(\frac{e^{-\frac{1}{2\sigma ^2(T_0-t)}\left(\ln\frac{x}{z_0}+\left(r-q-\frac{\sigma ^2}{2}+\alpha \sigma^2\right)(T_0-t)\right)^2}e^{-\frac{1}{2\sigma ^2(T_1-T_0)}\left(\ln\frac{x}{z_0}+\left(r-q-\frac{\sigma ^2}{2}+\alpha \sigma^2\right)(T_1-T_0)\right)^2}}{2\pi \sigma ^2z_0z_1\sqrt{(T_0-t)(T_1-T_0)}}\right)
\\
~~~~~~~~~~~~~~1(s_0z_0>s_0\xi _0)1(s_1z_1>s_1\xi _1)
\end{multline}
In this integral, we use the change of variables
$$y_i=\left(\ln\frac{x}{z_i}+\left(r-q-\frac{\sigma ^2}{2}+\alpha \sigma ^2\right)(T_i-t)\right)(\sigma \sqrt{T_i-t})^{-1},i=0,1$$

Then  $dy_i=(\sigma \sqrt{T_i-t})^{-1}\frac{-dz_i}{z_i}$. Considering
$$-\frac{y_0^2}{2}-\frac{(y_0\sigma \sqrt{T_0-t}-y_1\sigma \sqrt{T_1-t})^2}{2\sigma ^2(T_1-T_0)}=\frac{y_0^2-2\rho y_0y_1+y_1^2}{2(1-\rho ^2)},\rho=\sqrt{\frac{T_0-t}{T_1-t}}$$
we have
\begin{multline}
(M^\alpha)^{s_0 s_1}_{\xi _0 \xi _1}(x,t)=\nonumber
\\
=x^\alpha e^{\mu (T_1-t)}\frac{1}{2\pi \sqrt{1-\rho ^2}}\int_{-\infty}^{+\infty}\int_{-\infty}^{+\infty}e^{\frac{y_0^2-2\rho y_0y_1+y_1^2}{2(1-\rho ^2)}}1(s_0y_0<s_0d_0)1(s_1y_1<s_1d_1)dy_0dy_1=I
\end{multline}
We use the following change of variables again.
$$y'_0=s_0y_0,y'_1=s_1y_1$$
The Jacobians of this change is 
\begin{equation*} 
\left| \det
\begin{pmatrix}
s_0 & 0\\
0 & s_1
\end{pmatrix}
\right| =1
\end{equation*}
so we have
\begin{multline}
(M^\alpha)^{s_0 s_1}_{\xi _0 \xi _1}(x,t)=\nonumber
\\
=x^\alpha e^{\mu (T_1-t)}\frac{1}{2\pi \sqrt{1-\rho ^2}}\int_{-\infty}^{+\infty}\int_{-\infty}^{+\infty}e^{\frac{{y'}_0^2-2s_0s_1\rho y'_0y'_1+{y'}_1^2}{2(1-\rho ^2)}}1(s_0y'_0<s_0d_0)1(s_1y'_1<s_1d_1)dy'_0dy'_1
\\
=x^\alpha  e^{\mu (T_1-t)}\frac{1}{2\pi \sqrt{1-\rho ^2}}N_2(s_0d_0,s_1d_1;s_0s_1\rho)
\end{multline}

(QED).
\end{proof}
\begin{definition}
  An \textbf{n-th $\alpha$-power binary option} is inductively defined as the binary contract with the expiry date $T_0$  underlying on an (n-1)-th order $\alpha$-power binary option. That is, \textbf{an n-th order$\alpha$-power binary option}'s price is the solution of \eqref{eq-1} with the time  $T_0$ payoff of the following form 
\begin{equation}\label{eq-41}
V(x,T_0)=(M^\alpha)^{s_1\cdots s_{n-1}}_{\xi _1 \cdots \xi _{n-1}}(x,T_0;T_1,\cdots ,T_{n-1})1(s_0x>s_0\xi _0)
\end{equation}
Here $(M^\alpha)^{s_1\cdots s_{n-1}}_{\xi _1 \cdots \xi _{n-1}}(x,T_0;T_1,\cdots ,T_{n-1}$  is the price of (n-1)-th  order $\alpha$-power binary option with expiry dates $T_1,\cdots,T_{n-1}$.The price of this n-th order $\alpha $-power binary option is denoted by $(M^\alpha)^{s_0,s_1\cdots s_{n-1}}_{\xi _0,\xi _1 \cdots \xi _{n-1}}(x,t;T_0,T_1,\cdots ,T_{n-1})$  and this option is called an n-th order $\alpha$-power binary option with n expiry dates $T_0,T_1,\cdots ,T_{n-1}$.
\end{definition}
Between the prices of $n$-th order $\alpha$-power binary options and the corresponding price of $(n-1)$-th order  $\alpha$-power binary option, there is a following parity relation 
$$(M^\alpha)^{+,s_1\cdots s_{n-1}}_{\xi _0,\xi _1 \cdots \xi _{n-1}}(x,t)+(M^\alpha)^{-,s_1\cdots s_{n-1}}_{\xi _0,\xi _1 \cdots \xi _{n-1}}(x,t)=(M^\alpha)^{s_1\cdots s_{n-1}}_{\xi _0,\xi _1 \cdots \xi _{n-1}}(x,t),t<T_0$$

\begin{theorem}\label{theorem 6}
The price of n-th  $\alpha$-power binary option (the solution to the problem (\eqref{eq-1},\eqref{eq-41}) is given as following.
$$(M^\alpha)^{s_0,s_1\cdots s_{n-1}}_{\xi _0,\xi _1 \cdots \xi _{n-1}}(x,t;T_0,T_1,\cdots ,T_{n-1})=x^\alpha e^{\mu (T_{n-1}-t)}N_n(s_0d_0,\cdots,s_{n-1}d_{n-1};A_n(s_0\cdots s_{n-1}))$$
Here
$$d_i=\left(\ln\frac{x}{\xi _i}+\left(r-q-\frac{\sigma ^2}{2}+\alpha \sigma ^2\right)(T_i-t)\right)(\sigma \sqrt{T_i-t})^{-1}$$
$$N_n(d_0,\cdots,d_{n-1};A)=\left(\frac{1}{2\pi}\right)^{\frac{n}{2}}\frac{1}{|detA|^{\frac{1}{2}}}\int_{-\infty}^{d_0}\int_{-\infty}^{d_{n-1}}e^{-\textbf{y}^TA^{-1}\textbf{y}}dy_0\cdots dy_{n-1}$$
$$A(s_0s_1\cdots s_{n-1})=(s_is_ja_{ij})^{n-1}_{i,j=0},det[A(s_0s_1\cdots s_{n-1})]=det(A)$$
And $A=(a_{ij})_{i,j=0,\cdots ,n-1}$is n-th matrix which is defined as following like \cite{OK}.
$$a_{00}=(T_1-t)/(T_1-T_0),$$
$$a_{n-1,n-1}=(T_{n-1}-t)/(T_{n-1}-T_{n-2}),$$
$$a_{ii}=(T_i-t)/(T_i-T_{i-1})+(T_i-t)/(T_{i+1}-T_i)$$
$$a_{i,i+1}=a_{i+1,i}=-\sqrt{(T_i-t)(T_{i+1}-t)}/(T_{i+1}-T_i),0\le i \le n-2$$
And for other indices,$a_{ij}=0.$
\end{theorem}
\begin{proof}
The cases of $n=1$ and $n=2$ were proved by Theorem \ref{theorem 2} and Theorem \ref{theorem 5}. In the case of $n>2$ we will give a sketch of the proof by induction.

We assume that Theorem \ref{theorem 6} holds in the case of $n-1$. From the definition 4.2, $(M^\alpha)^{s_0,s_1\cdots s_{n-1}}_{\xi _0,\xi _1 \cdots \xi _{n-1}}(x,t;T_0,T_1,\cdots ,T_{n-1})$ satisfies \eqref{eq-40} and
$$V(x,T_0)=(M^\alpha)^{s_1\cdots s_{n-1}}_{\xi _1 \cdots \xi _{n-1}}(x,T_0;T_1,\cdots ,T_{n-1})\cdot 1(s_0x>s_0\xi_0)$$
Here $(M^\alpha)^{s_1\cdots s_{n-1}}_{\xi _1 \cdots \xi _{n-1}}(x,T_0T_1,\cdots ,T_{n-1})$ is the price of the (n-1)-th power binary option.
Therefore by preposition 1 of \cite{OK}
\begin{multline}
(M^\alpha)^{s_0s_1\cdots s_{n-1}}_{\xi _0\xi _1 \cdots \xi _{n-1}}(x,t;T_1,\cdots ,T_{n-1})=\nonumber
\\
=e^{-r(T_0-t)}\int_{-\infty}^{+\infty}\frac{1(s_0z>s_0\xi _0)}{\sigma \sqrt{2\pi (T_0-t)}}\frac{1}{z}e^{-\frac{\left(\ln\frac{x}{z}+\left(r-q-\frac{\sigma ^2}{2}\right)(T_i-t)\right)^2}{2\sigma ^2(T_0-t)}}(M^\alpha)^{s_1\cdots s_{n-1}}_{\xi _1\cdots \xi _{n-1}}(z,T_0)dz
\end{multline}
By induction-assumption, the result of Theorem \ref{theorem 6} holds for $(M^\alpha)^{s_1\cdots s_{n-1}}_{\xi _1\cdots \xi _{n-1}}(z,T_0)$. Thus we have 
$$(M^\alpha)^{s_1\cdots s_{n-1}}_{\xi _1\cdots \xi _{n-1}}(z,T_0)=x^\alpha e^{\mu (T_{n-1}-T_0)}N_n(s_0d_0,\cdots,s_{n-1}d_{n-1};A_n(s_0\cdots s_{n-1}))$$
Substitute this equality into the above integral representation and calculate the integral, then we have the result of Theorem \ref{theorem 6} for the case $n>2$.(Proof End)

\end{proof}


\begin{thebibliography}{99}

\bibitem{BB}
Benninga, S., Björk, T. and Wiener, Z. On the use of numeraires in option pricing. The Journal of Derivatives. Winter 2002. 10(2): 1-16.
\bibitem{P1}
P. Buchen., The Pricing of dual expiry exotics, Quantitative Finance, 4, 2004, 101-108.
\bibitem{P2}
P. Buchen (2012) An Introduction to Exotic Option Pricing, CRC Press, 107-124.
\bibitem{MS}
M. Garman, S. Kohlhagen, ''Foreign currency option values'' , Journal of international money and Finance, 2(1983), 231-237.
\bibitem{I}
Ingersoll, J. E., Digital contract: simple tools for pricing complex derivatives, J. Business, 73 (2000), 67-88
\bibitem{J}
Jiang, Li Shang(2005) : Mathematical Modeling and Method of Option Pricing, World scientific, Singapore, 200~249.
\bibitem{RR} 
Rubinstein, M. and Reiner, E., Unscrambling the binary code, Risk Mag. 4 (1991) 75-83
\bibitem{OK}
O, Hyong-chol and Kim, Mun-chol; Higher Order Binary Options and Multiple Expiry Exotics, Electronic Journal of Mathematical Analysis and Applications, Vol. 1(2) July 2013, 247-259.
\bibitem{ORW} 
O, Hyong-chol, RO, Yong-hwa, Wan, Ning, The Use of Numeraires in Multi-dimensional Black-Scholes Partial Differential Equations, arXiv1310.8296v3 [q-fin.PR], 1~23.
\bibitem{ODC} 
O, Hyong-chol,, D.H. Kim and C.H. Pak, Analytical pricing of defaultable discrete coupon bonds in unified two-factor model of structural and reduced form models, J. Math. Anal. Appl.416 (2014), no. 1, 314-334.
\bibitem{ODJ} 
O, Hyong-chol,, D.H. Kim, J.J. Jo and S.H. Ri, Integrals of higher binary options and defaultable bonds with discrete default information, Electron. J. Math. Anal. Appl. 2 (2014), no. 1, 190-214.
\bibitem{OYD} 
O, Hyong-chol, Yong-Gon Kim , D-H. Kim, Higher binary with time dependent coefficients and 2 factor model for Defaultable Bond with Discrete Default Information, Malaya Journal of Matematik, 2(4), 2014, pp 330-344.
\bibitem{OJJ} 
O, Hyong-chol, J.J.Jo and J.S. Kim, General properties of solutions to inhomogeneous Black-Scholes equations with discontinuous maturity payoffs, J. Differential Equations 260 (2016), no. 4, 3151-3172.
\bibitem{OJS} 
O, Hyong-chol, J.J.Jo, S.Y. Kim and S.G. Jon; A comprehensive unified model of structural and reduced form type for defaultable fixed-income bonds, Bulletin of the Iranian Mathematical Society, Vol. 43, No 3, 575-599, 2017.
\bibitem{SB} 
Skipper M., P.Buchen, A Valuation Formula for Multi-asset, Multi-period Binaries in a Black–Scholes Economy, ANZIAM J. 50(2009), 475–485

\end{thebibliography}
\end{document}